\documentclass[11pt]{article}
\usepackage{graphicx}
\usepackage{fullpage}
\usepackage{amsmath,amsthm,amssymb}
\usepackage{amsmath,amssymb}
\usepackage{color}
\usepackage{verbatim}
\usepackage{microtype}
\usepackage{algorithm}
\usepackage{algpseudocode}

\begin{document}

\def\lf{\left\lfloor}   
\def\rf{\right\rfloor}
\def\lc{\left\lceil}   
\def\rc{\right\rceil}

\newcommand{\spcut}{Sparsest-Cut }
\newtheorem{theorem}{Theorem}
\newtheorem{proposition}[theorem]{Proposition}
\newtheorem{claim}[theorem]{Claim}
\newtheorem{lemma}[theorem]{Lemma}
\newtheorem{corollary}[theorem]{Corollary}

\newcommand{\tw}{\mathsf{tw}}
\newcommand{\dem}{\mathsf{dem}}
\newcommand{\OPT}{\mathsf{OPT}}

\title{Quasimetric embeddings and their applications
\footnote{An extended abstract of this work appeared in ICALP 2016.}
\footnote{This work was supported by NSF grants CCF-1423230, CCF-1526513, IIS-1422400, and award CAREER-1453472.}}

\date{}

\author{Facundo M\'{e}moli\\
Dept. of Computer Science and Engineering and Dept.~of Mathematics\\
The Ohio State University\\
  Columbus OH, USA\\
  \texttt{memoli@math.osu.edu}
\and Anastasios Sidiropoulos \\
Dept. of Computer Science and Engineering and Dept.~of Mathematics\\
The Ohio State University\\
  Columbus OH, USA\\
  \texttt{sidiropoulos.1@osu.edu}
\and Vijay Sridhar\\
Dept. of Computer Science and Engineering\\
The Ohio State University\\
  Columbus OH, USA\\
  \texttt{sridhar.38@buckeyemail.osu.edu}}

\maketitle

\begin{abstract}
We study generalizations of classical metric embedding results to the case of \emph{quasimetric} spaces; that is, spaces that do not necessarily satisfy symmetry.
Quasimetric spaces arise naturally from the shortest-path distances on directed graphs.
Perhaps surprisingly, very little is known about low-distortion embeddings for quasimetric spaces.

Random embeddings into \emph{ultrametric} spaces are arguably one of the most successful geometric tools in the context of algorithm design.
We extend this to the quasimetric case as follows.
We show that any $n$-point quasimetric space supported on a graph of treewidth $t$ admits a random embedding into \emph{quasiultrametric} spaces with distortion $O(t \log^2 n)$, where  quasiultrametrics are a natural generalization of ultrametrics.
This result allows us to obtain  $t\log^{O(1)} n$-approximation algorithms for the Directed Non-Bipartite Sparsest-Cut and the Directed Multicut problems on $n$-vertex graphs of treewidth $t$, with running time polynomial in  both $n$ and $t$.

The above results are obtained by considering a generalization of random partitions to the quasimetric case, which we refer to as \emph{random quasipartitions}.
Using this definition and a construction of [Chuzhoy and Khanna 2009] we derive a polynomial lower bound on the distortion of random embeddings of general quasimetric spaces into quasiultrametric spaces.
Finally, we establish a lower bound for embedding the shortest-path quasimetric of a graph $G$ into graphs that exclude $G$ as a minor.
This lower bound is used to show that several embedding results from the metric case do not have natural analogues in the quasimetric setting.
\end{abstract}

\section{Introduction}

Low-distortion embeddings of metric spaces have become an indispensable tool in the design of algorithms \cite{linial1995geometry,arora2008euclidean,bartal1998approximating,indyk20048}. 
We consider generalizations of some fundamental metric embedding problems to the case of \emph{quasimetric spaces}.
Formally, a quasimetric space is a pair $(X,d)$ where $X$ is the set of points and $d:X\times X\to \mathbb{R}_{+}\cup\{+\infty\}$, satisfying the following conditions:

\begin{description}
\item{(C1)}
For all $x,y\in X$, $d(x,y)=0$ iff $x=y$.
\item{(C2)}
For all $x,y,z\in X$, $d(x,y)\leq d(x,z)+d(z,y)$.
\end{description}

In other words, a quasimetric space satisfies all the conditions of a metric space, except for the following symmetry condition:

\begin{description}
\item{(C3)}
For all $x,y\in X$, $d(x,y)=d(y,x)$.
\end{description}

Finite quasimetric spaces are precisely the  shortest-path distances of finite directed graphs.
Perhaps surprisingly, many basic questions regarding low-distortion embeddings of quasimetric spaces are poorly understood.
As we explain below, these problems are of importance in algorithm design, and are
intricately tied to the approximability of various cut problems on directed edge-capacitated graphs.

\subsection{Our contributions}
We now outline our main results, and contrast with what was previously known for the case of metric spaces.
We consider quasimetric spaces that arise from the shortest-path distances of directed graphs.
For a family ${\cal F}$ of undirected graphs, we consider the directed graphs arising from the graphs in ${\cal F}$ by replacing every undirected edge $\{u,v\}$ by two edges $(u,v)$ and $(v,u)$ with opposite directions, and by assigning arbitrary positive edge lengths to them.
When ${\cal F}$ is the family of all trees (resp.~graphs of treewidth-$t$), we refer to the resulting family of quasimetric spaces as \emph{tree quasimetric spaces} (resp.~\emph{treewidth-$t$ quasimetric spaces}).

\paragraph*{Random embeddings}
A very successful metric embedding tool in the context of algorithm design is random embeddings.
The high-level idea is that given some ``complicated'' space, we can find a random embedding into some ``simpler'' space, preserving all distances in expectation.

Formally, let $M=(X,d)$ be a quasimetric space.
A \emph{random embedding} of $M$ is a distribution ${\cal F}$ over pairs $(f, M')$ where $M'=(X',d')$ is a quasimetric space and $f:X\to X'$, such that for any $x,y\in X$
\[
\Pr[d'(f(x), f(y)) \geq d(x,y)] = 1.
\]
Let $\alpha\geq 1$.
We say that the random embedding ${\cal F}$ has \emph{distortion} $\alpha$ if for all $x,y\in X$
\[
\mathbb{E} \left[ d'(f(x), f(y)) \right] \leq \alpha\, d(x,y).
\]

This definition is of algorithmic interest because for several optimization problems it allows us to reduce instances on general graphs to instances on simpler graphs (see \cite{bartal1996probabilistic,sidiropoulos2010optimal} for a more detailed exposition).
It has been shown that any $n$-point metric space admits a random embedding into \emph{ultrametric} spaces with distortion $O(\log n)$ \cite{fakcharoenphol2003tight} (see also \cite{bartal1996probabilistic,bartal1998approximating,abraham2008nearly}).
Here, an ultrametric space is a metric space that satisfies the following stronger version of the triangle inequality:

\begin{description}
\item{(C2*)}
For all $x,y,z\in X$, $d(x,y)\leq \max\{d(x,z),d(z,y)\}$.
\end{description}

It is easy to construct examples of quasimetric spaces that do not admit random embeddings into ultrametric spaces with bounded distortion.
This motivates the study of random embeddings of quasimetric spaces into \emph{quasiultrametric} spaces; these are precisely the quasimetric spaces that satisfy conditions (C1) \& (C2*).
We show that any $n$-point treewidth-$t$ quasimetric space admits a random embedding into quasiultrametric spaces with distortion $O(t \log^2 n)$.
In a similar fashion, we show that any treewidth-$t$ quasiultrametric admits an embedding into a convex combination of 0-1 quasimetric spaces with distortion $O(t \log^2 n)$; here, a 0-1 quasimetric space requires that all distances are either 0 or 1.
As we explain below, this result allows us to obtain new approximation algorithms for directed cut problems on treewidth-$t$ graphs.

\paragraph*{Random quasipartitions}
A fundamental primitive underlying many metric embedding results is \emph{random partitions} \cite{bartal1996probabilistic}.
This primitive has been successfully used in many diverse prolems \cite{klein1993excluded,krauthgamer2005measured,kelner2011metric,lee2005extending,calinescu2005approximation,lee2009geometry,lee2010genus}.
It is easy to construct examples of quasimetric spaces that do not admit good random partitions.
We overcome this technical obstacle by defining a \emph{quasipartition} to be a transitive reflexive relation.
This is a generalization of a partition for the following reason:
For a partition $P$ of some set $X$ we can define the relation $R$ on $X$ where for all $x,y\in X$, we set $(x,y)\in R$ iff $x$ and $y$ are in the same cluster in $P$.
It is easy to check that $R$ is indeed transitive and reflexive; however, there are transitive and reflexive relations that do not arise in this fashion.

Let $r\geq 0$.
We say that a quasipartition $P$ of $M$ is \emph{$r$-bounded} if for any $x,y\in X$, if $(x,y)\in P$, then $d(x,y)\leq r$.
Let ${\cal D}$ be a distribution over $r$-bounded quasipartitions of $M$.
We say that ${\cal D}$ is $r$-bounded.
We also say that ${\cal D}$ is \emph{$\beta$-Lipschitz}, for some $\beta>0$, if for any $x,y\in X$, we have that 
\[
\Pr_{P\sim {\cal D}}[(x,y)\notin P] \leq \beta \frac{d(x,y)}{r}.
\]
Given a distribution $\mathcal{D}$ over quasipartitions we ocasionally refer to any quasipartition $P$ sampled from $\mathcal{D}$ as a random quasipartition (with distribution $\mathcal{D}$).

We show that for all $r>0$, any tree quasimetric space admits a $O(1)$-Lipschitz, $r$-bounded random quasipartition.
We remark that no such result is possible using random partitions.
We further show that for all $r>0$, any treewidth-$t$ quasimetric space admits a $O(t \log n)$-Lipschitz, $r$-bounded random quasipartition.
This random quasipartition is at the heart of the random embedding result outlined above.

Using a result of Chuzhoy and Khanna \cite{directedcut} we show that the polynomial dependence on the treewidth is necessary.
More precisely, there exist $n$-point quasimetric spaces that do not admit $o(n^{1/7}/\log^{4/7} n)$-Lipschitz quasipartitions.
Using this lower bound we further show that there exist $n$-point quasimetrics that do not admit random embeddings into quasiultrametric spaces with distortion $o(n^{1/7}/\log^{4/7})$.

\paragraph*{Applications to cut problems on directed graphs}
Using the above result for embedding treewidth-$t$ quasimetric spaces into a convex combination of 0-1 quasimetric spaces, we show that the integrality gap of the Directed Non-Bipartite \spcut LP on graphs of treewidth $t$ it $O(t \log^2 n)$.
This implies a $O(t  \sqrt{\log t} \log^2 n)$-approximation algorithm for the Directed \spcut problem with running time polynomial in both $n$ and $t$.
We remark that dynamic-programming based techniques for graphs of bounded treewidth can only yield algorithms with running time exponential in $t$, and are thus practical only for very small values of $t$.
Our result provides an interesting trade-off between running time and approximation guarantee for graphs of moderately large treewidth.
For example, our result implies a polynomial time $\log^{O(1)}$-approximation for Directed Non-Bipartite Sparsest-Cut on graphs of treewidth $\log^{O(1)} n$.

Similarly, we obtain a $O(t  \sqrt{\log t} \log^3 n)$-approximation algorithm for the Directed Multicut problem on graphs of treewidth $t$, with running time polynomial in both $n$ and $t$.

\paragraph*{Lower bound for random topological simplification}
It has been show that for various classes of topologically restricted graphs, there exist constant-distortion random embeddings into topologically simpler graphs.
For example, graphs of bounded genus admit constant-distortion random embeddings into planar graphs \cite{indyk2007probabilistic,borradaile2009randomly,sidiropoulos2010optimal}, and similar results are known for more general classes of minor-free graphs \cite{lee2009geometry}.

We show that no such result is possible for the case of directed graphs.
More precisely, we show that for any directed acyclic graph $G$, there exists a subdivision $G'$ of $G$, such that for any embedding of the shortest-path quasimetric of $G'$ into the shortest-path quasimetric of some graph $H$ with bounded distortion, we have that $G$ is a minor of $H$.
For example, this implies that there is no  bounded-distortion random embedding of toroidal (i.e.~genus-1) quasimetric spaces into planar quasimetric spaces.

\section{Quasipartitions of tree quasimetrics}\label{sec:tree}
In this section we describe a method to construct an $O(1)$-Lipschitz distribution over $r$-bounded quasipartitions of tree quasimetric spaces.
More precisely, we prove the following result.

\begin{theorem}\label{thm:trees}
Let $M$ be a shortest path quasimetric space supported on some directed tree $T$ in which every edge has non-negative weights in both directions. For any $r>0$, there exists an $O(1)$-Lipschitz distribution over $r$-bounded quasipartitions of $M$.
\end{theorem}  

First we describe an algorithm to construct a distribution over $r$-bounded quasipartitions of $M$. We will then show that the distribution produced by the algorithm is $O(1)$-Lipschitz.

\begin{algorithm}[htbp]
\caption{Random quasipartition of a tree quasimetric space}
\begin{algorithmic}
\Require A tree quasimetric space $M=(V(T),d_M)$ and $r>0$.
\Ensure An $r$-bounded probabilistic quasipartition $R$. \\
\begin{description}
\item{Step 1.} Set $R=\emptyset$. Add $(u,v)$ to $R$ for all $(u,v) \in E(M)$. 
\item{Step 2.} Pick a root vertex $t$ in $M$ arbitrarily.
\item{Step 3.} Pick $z \in  [0,r/2]$ uniformly at random.
\item{Step 4.} For all $(u,v) \in E(M)$ remove $(u,v)$ from $R$ if at least one of the following holds:
\begin{description}
\item{(a)}
$d_{M}(u,t) > z + i  \frac{r}{2}$ and $d_{M}(v,t) \leq z + i  \frac{r}{2}$ for any integer $i>0$.
\item{(b)}
$d_{M}(t,v) > z + i  \frac{r}{2}$ and $d_{M}(t,u) \leq z + i  \frac{r}{2}$ for any integer $i>0$.
\end{description}
\item{Step 5.} Enforce transitivity on R: For all $u,v,w \in V(M)$, if $(u,v) \in R$ and $(v,w) \in R$, then add $(u,w)$ to $R$.
\end{description}
\end{algorithmic}
\end{algorithm}

We shall now prove some properties of the random quasipartition $R$ in the following Lemmas. We will use these to prove Theorem \ref{thm:trees}.

\begin{lemma}\label{lemtreeproperty}
Let $u,v\in V(T)$.
Let $D$ be the path from $u$ to $v$ in $T$. Then either $u$ is in the path from $t$ to $v$, or $v$ is in the path from $u$ to $t$, or there exists a vertex $w$ on $D$ such that $w$ lies on the path from $u$ to $t$ and $w$ lies on the  path from $t$ to $v$.
\end{lemma}

\begin{proof}
Let $w$ be the nearest common ancestor of $u$ and $v$ in $T$.
If $w=u$, then $u$ is in the path from $t$ to $v$; if $w=v$, then $v$ is in the path from $t$ to $u$; if $w\neq u$ and $w\neq v$, then $w$ is in the path from $t$ to $u$ and in the path from $t$ to $v$.
\end{proof}

\begin{lemma}\label{lemlayering}
For any $u,v \in V(T)$ where $u$ is in the path from $t$ to $v$, let $P=\{a_1 = u,a_2,\ldots,a_m = v\}$ be the path from $u$ to $v$. If $(a_i,a_{i+1}) \in R$ for all $i\in \{1,\ldots,m-1\}$ after Step 4 then $d_M(u,v) \leq \frac{r}{2}$.
\end{lemma}

\begin{proof}
Let $j$ be the largest integer such that $d_{M}(t,a_1) > z + j \frac{r}{2}$. By the choice of $j$ it must be that $z + j  \frac{r}{2} < d_{M}(t,a_1) \leq z + (j+1) \frac{r}{2}$. Since $(a_i,a_{i+1})$ is not removed from $R$ in Step 4 of the algorithm it must be that $z + j \frac{r}{2} < d_{M}(t,a_i) \leq z + (j+1) \frac{r}{2}$ for all $i\in \{1,\ldots,m\}$. This implies that $d_{M}(t,v) \leq z + (j+1) \frac{r}{2} \leq d_{M}(t,u) + \frac{r}{2}$. Since $d_M(t,v) = d_M(t,u) + d_M(u,v)$, we have that $d_M(u,v) \leq \frac{r}{2}$, which concludes the proof.
\end{proof}

\begin{lemma}\label{lemlayering2}
For any $u,v \in V(M)$ where $v$ is in the path from $u$ to $t$, let $P=\{a_1 = u,a_2,\ldots,a_m = v\}$ be the path from $u$ to $v$. If $(a_i,a_{i+1}) \in R$ for all $i\in \{1,\ldots,m-1\}$ after Step 4 then $d_M(u,v) \leq r/2$.
\end{lemma}

\begin{proof}
The proof is similar to the proof of Lemma \ref{lemlayering}.
\end{proof}

\begin{lemma}\label{rbounded}
If $(u,v) \in R$ then $d_M(u,v) \leq r$.
\end{lemma}

\begin{proof}
The fact that $(u,v) \in R$ implies that at the beginning of Step 4 there must have been a path $P=\{a_1 = u,a_2,\ldots,a_m = v\}$ from $u$ to $v$ such that $(a_i,a_{i+1}) \in R$ for all $i\in \{1,\ldots,m-1\}$. Since $M$ is a tree quasimetric space, the shortest path is the single unique path from $u$ to $v$ for any $u,v \in V(T)$. From Lemma \ref{lemtreeproperty} we have that one of the following three cases is true:

\emph{Case 1:} $u$ is in the shortest path from $t$ to $v$. We have $d_{M}(u,v) \leq r/2$ from Lemma \ref{lemlayering}.

\emph{Case 2:} $v$ is in the shortest path from $u$ to $t$. We have $d_{M}(u,v) \leq r/2$ from Lemma \ref{lemlayering2}.

\emph{Case 3:} There exists $a_{j}$ that lies on the shortest path from $u$ to $t$ and on the shortest path from $t$ to $v$. From Lemmas \ref{lemlayering} and \ref{lemlayering2} we have that $d_{M}(u,a_j) \leq r/2$ and $d_{M}(a_j,v) \leq r/2$. By the triangle inequality we get $d_{M}(u,v) \leq r$.
\end{proof}

\begin{lemma}\label{edgeprob}
Any $(u,v) \in E(T)$ is removed with probability at most $2 d_M(u,v)/r$ in Step 4 of the algorithm.
\end{lemma}

\begin{proof}	
Since $M$ is a tree quasimetric space there are exactly two cases:

\emph{Case 1:} The edge $(u,v)$ is in the direction away from $t$ which implies that $d_M(t,u) \leq d_M(t,v)$. Let $i$ be the largest integer such that $i\cdot r/2 \leq d_M(t,u)$. The edge $(u,v)$ is removed from $R$ if $z$ is chosen between $d_{M}(t,u) -i  \cdot r/2$ and $d_{M}(t,v) - i \cdot r/2$. The probability of that event is bounded by $\int_{d_{M}(t,u) -i  \frac{r}{2} }^{d_{M}(t,v) -i  \frac{r}{2}} p(z) dz = \frac{2}{r}  (d_{M}(t,v) - d_{M}(t,u)) \leq 2 d_M(u,v)/r$ by the  triangle inequality.

\emph{Case 2:} The edge $(u,v)$ is in the direction toward $t$ which implies that $d_M(v,t) \leq d_M(u,t)$. Let $i$ be the largest integer such that $i\cdot r/2 \leq d_M(v,t)$ . $(u,v)$ is removed from $R$ if $z$ is chosen between $d_{M}(v,t) -i  \cdot r/2$ and $d_{M}(u,t) - i  \cdot r/2$. The probability of that event is bounded by $\int_{d_{M}(v,t) -i \cdot r/2}^{d_{M}(u,t) -i \cdot r/2} p(z) dz = \frac{2}{r}  (d_{M}(u,t) - d_{M}(v,t)) \leq 2 d_M(u,v)/r$ by the triangle inequality.
\end{proof}

\begin{lemma}\label{probabilityofcutting}
 $\Pr[(u,v) \not \in R] \leq 2 d_M(u,v)/r$ for all $u,v \in V(T)$.
\end{lemma}

\begin{proof}	
Let the unique path from $u$ to $v$ in $M$ be $p= \{x_1 = u,x_2,\ldots,x_h = v\}$. Let $X_{p}$ be the event that path $p$ contains at least one edge $(x_i,x_{i+1})$ such that $(x_i,x_{i+1}) \not \in R$ at the beginning of Step 5. Let $Y_{(a,b)}$ be the event that $(a,b) \not \in R$ for $(a,b) \in E(G)$. 
We have
$\Pr [X_{p}] = \Pr[Y_{(x_1,x_2)} \vee \ldots \vee Y_{(x_{h-1},x_h)} ]$. From Lemma \ref{edgeprob} and the union bound we have that $\Pr[Y_{(x_1,x_2)} \vee \ldots \vee Y_{(x_{h-1},x_h)} ] \leq \Pr[Y_{(x_1,x_2)}] +  \ldots + \Pr[Y_{(x_{h-1},x_h)}] 
\leq
2d_M(x_1,x_2)/r + \ldots + 2d_M(x_{h-1},x_h)/r
= 2  d_M(u,v)/r$, concluding the proof.
\end{proof}	
We are now ready to prove the main result of this Section.
\begin{proof}[Proof of Theorem \ref{thm:trees}]
It follows by Lemmas \ref{rbounded} and \ref{probabilityofcutting} that the algorithm outputs an $O(1)$-Lipschitz distribution over $r$-bounded quasipartitions of $M$.
\end{proof}

\section{Quasipartitions for graphs of small treewidth}\label{sec:treewidth}
In this section we prove the existence of a $O( t \log{n} )$-Lipschitz distribution over $r$-bounded quasipartitions for any quasimetric supported on a directed graph of treewidth $t$. The main result is summarized in the following.

\begin{theorem}\label{treewidth quasipartitions result}
Let $G$ be a $n$-vertex directed graph of treewidth $t$. Let $M$ be the shortest-path quasimetric space induced by $G$. Then for any $r>0$, there exists an $O(t \log{n} )$-Lipschitz distribution over $r$-bounded quasipartitions of $M$.
\end{theorem} 

In the proof of the above theorem we use the following proposition which is immediate from the definition of treewidth.

\begin{proposition}\label{prop: balanced separators tw}
Any graph $G$ of treewidth $t$ has a set of vertices $K \subseteq V(G)$ where $|K| \leq t$ such that removing $K$ gives connected components each of which contains at most $\frac{|V(G)|}{2}$ vertices.
\end{proposition}

First we introduce an algorithm to construct the required distribution over $r$-bounded quasipartitions of $M$. Steps 2 to 4 of the algorithm are recursive. At each recursive call the algorithm works on an associated sub-graph $G^*$ and a global set $R$ which is common to all recursive calls.

\begin{algorithm}[htbp]
\caption{Random quasipartition of a bounded treewidth graph}
\begin{algorithmic}
\Require A digraph $G$ of treewidth $t$, and $r>0$.
\Ensure A random $r$-bounded quasipartition $R$. \\
Initialization: 
Set $G^*=G$ and $R=E(G)$.
Perform the following recursive algorithm on $G^*$.
\begin{description}
\item{Step 1.} 
Pick $z \in [0,r/2]$ uniformly at random.
\item{Step 2.} If $|V(G^*)| \leq 1$, terminate the current recursive call. Otherwise pick a set of vertices $K \subseteq V(G^*)$ such that $|K| \leq t$ and removing $K$ from $G^*$ gives connected components $C_1,\ldots,C_m$, each containing at most $\frac{|V(G^*)|}{2}$ vertices. This is possible by Proposition \ref{prop: balanced separators tw}.
\item{Step 3.} For all $(u,v) \in E(G^*)$ remove $(u,v)$ from R if one of the following holds:

\begin{description}
\item{(a)}
$d_{G}(u,x)>z$ and $d_{G}(v,x) \leq z$ for some vertex $x \in K$.

\item{(b)}
$d_{G}(x,v)>z$ and $d_{G}(x,u) \leq z$ for some vertex $x \in K$.\end{description}

\item{Step 4.} Recursively call Steps 2-4 on the vertex-induced subgraphs $G^*[C_1],\ldots,G^*[C_m]$.

\item{Step 5.} Once all branches of the recursion terminate enforce transitivity on $R$: For all $u,v,w \in V(G)$ if $(u,v) \in R$ and $(v,w) \in R$ add $(u,w)$ to $R$.
\end{description}
\end{algorithmic}
\end{algorithm}

Next we state some properties of the resulting random quasipartition $R$. We will use these to prove the main theorem. 

\begin{lemma}\label{lem:r bounded}
If $(u,v) \in R$ then $d_G(u,v) \leq r$.
\end{lemma}

\begin{proof}
Suppose this were not the case and $d_G(u,v) > r$. The fact that $(u,v) \in R$ implies that at the beginning of Step 5 there must have been a path $P=\{a_1 = u,a_2,\ldots,a_m = v\}$ from $u$ to $v$ such that $(a_i,a_{i+1}) \in R$ for all $i\in [1,m-1]$. Consider the first branch of the recursion when a vertex on this path was part of $K$, the balanced separator. Let the vertex in $P$ that was part of K be $a_k$. Now we have that $d_{G}(a_k,a_k) = 0 \leq z$. This along with the fact that $(a_i,a_{i+1}) \in R$ for all $i\in [1,m-1]$ implies that $d_{G}(a_j,a_k) \leq z$ for all $a_j$ where $j<k$. Similarly it must be that $d_{G}(a_k,a_j) \leq z$ for all $a_j$ where $j>k$. Since $d_{G}(a_1,a_k) \leq z$ and $d_{G}(a_k,a_m) \leq z$ we have that $d_{G}(a_1,a_m)\leq 2z \leq r$ by triangle inequality. Since $d_G(u,v) = d_{G}(a_1,a_m) $ we have that $d_G(u,v) \leq r$ which is a contradiction.
\end{proof}

\begin{lemma}\label{depth}
The depth of the recursion is $O(\log{n})$.
\end{lemma}

\begin{proof}
This is because at every level of recursion the number of vertices in each component is at most half the number of vertices of the parent component in the previous level. 
\end{proof}

\begin{lemma}\label{edge}
Any $(u,v) \in E(G)$ is removed with probability at most $4 t \frac{d(u,v)}{r}$ in Step 4 of the algorithm.
\end{lemma}

\begin{proof}	
Consider a vertex $x \in K$. $(u,v)$ is removed from $R$ if $z$ is chosen between $d_{G}(u,x)$ and $d_{G}(v,x)$ or between $d_{G}(x,v)$ and $d_{G}(x,u)$. The probability of that event is bounded by $\int_{d_{G}(v,x)}^{d_{G}(u,x)} p(z) dz + \int_{d_{G}(x,u)}^{d_{G}(x,v)} p(z) dz = \frac{2}{r} (d_{G}(u,x) - d_{G}(v,x)) + \frac{2}{r} (d_{G}(x,v) - d_{G}(x,u)) \leq 4 \frac{d(u,v)}{r}$ from triangle inequality.

Taking union bound over all the vertices in K we get that $(u,v) \in E(G)$ is removed with probability at most $4 t   \frac{d(u,v)}{r}$ in Step 4 of the algorithm.
\end{proof}

\begin{lemma}\label{lem:logn lipschitz}
 $\Pr[(u,v) \not \in R] \leq 4 t \log{n}  \frac{d_G(u,v)}{r}$ for all $u,v \in V(G)$.
\end{lemma}

\begin{proof}
Let the set of all paths from $u$ to $v$ in $G$ be $P=\{p_1,p_2,\ldots,p_m\}$. Let $X_{p_i}$ be the event that path $p_i$ contains at least one edge $(w,z)$ such that $(w,z) \not \in R$ at the beginning of Step 6. Since we enforce transitivity in Step 6 of the algorithm we have that $\Pr[(u,v) \not \in R] = \Pr [X_{p_1} \wedge X_{p_2} \wedge \ldots \wedge X_{p_m} ]$. This implies that $\Pr[(u,v) \not \in R] \leq \Pr [X_{p_i}]$ for all $i \in [1,m]$. Now consider $p_j = \{x_1 = u,x_2,\ldots,x_h = v\}$, the shortest path from $u$ to $v$ in $G$. Let $Y_{(a,b)}$ be the event that $(a,b) \not \in R$ for $(a,b) \in E(G)$. We have $\Pr [X_{p_j}] = \Pr[Y_{(x_1,x_2)} \vee Y_{(x_2,x_3)} \vee \ldots \vee Y_{(x_{h-1},x_h)} ] \leq \Pr[Y_{(x_1,x_2)}] + \Pr[Y_{(x_2,x_3)}] + \ldots + \Pr[Y_{(x_{h-1},x_h)}] $ applying the union bound. Combined with Lemmas \ref{depth} and \ref{edge} this implies that $\Pr [X_{p_j}] \leq 4  t  \log{n}  \frac{d_G(p_j)}{r} $. Since $p_j$ is the shortest path from $u$ to $v$ we have that $4  t  \log{n}  \frac{d_G(p_j)}{r} = 4  t  \log{n}  \frac{d_G(u,v)}{r}$.
\end{proof}	

With these Lemmas we can now prove the main result of this section.

\begin{proof}[Proof of Theorem \ref{treewidth quasipartitions result}]
It follows by Lemmas \ref{lem:r bounded} and \ref{lem:logn lipschitz} that the Algorithm outputs an $O(t \log{n})$-Lipschitz distribution over $r$-bounded quasipartitions of $M$.
\end{proof}

The above algorithm can be implemented in polynomial time with an additional $O(\sqrt{\log t})$ loss on the quality of the partition.
This is summarized in the following Theorem.

\begin{theorem}\label{thm:tw poly algo}
Let $M$ be the shortest-path quasimetric space induced by a $n$-vertex directed graph $G$ of treewidth $t$. Then, there exists an algorithm with running time polynomial in $n$ and $t$ that computes the set of all quasipartitions in the support of an $O(t \sqrt{\log{t}}\log{n} )$-Lipschitz distribution over $r$-bounded quasipartitions of $M$, for any $r>0$.
\end{theorem}

\begin{proof}
The randomized algorithm described in Theorem \ref{treewidth quasipartitions result} can be derandomized to yield a 
polynomial time algorithm. First we note that it is possible to find in polynomial time a set $K$, with $|K|=O(t \sqrt{\log{t}})$  such that removing $K$ from $G$ gives connected components containing at most $\frac{2|V(G)|}{3}$ vertices \cite{feige2008improved}. We can use this in Step 2 of the algorithm.
The only random decision in the algorithm is in Step 1 when $z$ is chosen. Given $r > 0$, we can instead select $z$ exhaustively from all values in the set $S = \{ d(u,v) : u,v \in V(G) \text{ and } d(u,v)\leq \frac{r}{2} \}$. It can be observed from Step 3 of the algorithm that picking any other value of $z$ does not produce a new non-trivial $r$-bounded quasipartition. Since there are less than $n^2$ elements in $S$ this derandomized version of the algorithm runs in polynomial time in $n$ and the set of quasipartitions returned has less than $n^2$ elements.
\end{proof}

\section{Embeddings into quasiultrametrics and into convex combinations of 0-1 quasimetrics}\label{sec:convex}

In this Section we present our results on random embeddings into quasiultrametric spaces, and deterministic embeddings into convex combinations of 0-1 quasimetric spaces  (quasimetrics where all distances are either 0 or 1).

\subsection{Upper bounds}
We begin by establishing a relationship between quasipartitions and embeddings of quasimetric spaces into quasiultrametric spaces and 0-1 quasimetric spaces.
We say that a distribution over quasipartitions $\mathcal{D}$ is \emph{$\epsilon$-forcing}  if whenever $u,v\in X$ are such that $d(u,v) \leq \epsilon r$ then $\Pr_{P\sim {\cal D}}[(u,v)\notin P] = 0$.
First we state a result, inspired by \cite{bartal1996probabilistic}, that we use in subsequent proofs.

\begin{lemma}\label{epsilon forcing distributions}
Let $G$ be a directed graph on $n$ vertices. Let $M_{W}$ denote the shortest-path quasimetric space induced by $G$ where edge weights are specified by a function $W: E(G)\to \mathbb{R^{+}}$.
Suppose that for all $r>0$ there exists a $\beta$-Lipschitz distribution over $r$-bounded quasipartitions of $M_{W}$.
Then for all $r>0$ there exists a $2 \beta$-Lipschitz $\frac{1}{2n}$-forcing distribution over $r$-bounded quasipartitions of $M_{W}$.
\end{lemma}

\begin{proof}
Given $r>0$ and $W: E(G)\to \mathbb{R^{+}}$ we define $W': E(G)\to \mathbb{R^{+}}$ as follows: $W'(e) = W(e)$ if $W(e) > \frac{r}{2n}$ and $W'(e) = 0$ otherwise.

Consider ${\cal D}$, a $\beta$-Lipschitz distribution over $\frac{r}{2}$-bounded quasipartitions of $M_{W'}$. We have 
\[
 \Pr_{P \sim {\cal D}}[(u,v) \notin P] \leq \beta  \frac{d_{M_{W'}}(u,v)}{r/2}
\]
This implies that,
\[
 \Pr_{P \sim {\cal D}}[(u,v) \notin P] \leq 2 \beta  \frac{d_{M_W}(u,v)}{r}
\]
Since $W'(e) = 0$ if $W(e) \leq \frac{r}{2n}$ we have that
\[
 \Pr_{P \sim {\cal D}}[(u,v) \notin P] = 0
\]
for any $u,v \in V(G)$ where $d_{M_W}(u,v)\leq \frac{r}{2n}$. If $P$ is an $\frac{r}{2}$-bounded quasipartition of $M_{W'}$, then we claim that $P$ is also an $r$-bounded quasipartition of $M_{W}$. 
Since any shortest path in $G$ has at most $n$ edges we have that $d_{M_W}(u,v) \leq d_{M_{W'}}(u,v) + \frac{r}{2}$ for all $u,v \in V(G)$. Consider $P$ an $\frac{r}{2}$-bounded quasipartition of $M'$. For any $(u,v) \in P$ we have that $d_{M_{W'}}(u,v) \leq \frac{r}{2}$. Therefore we have that $d_{M_W}(u,v) \leq \frac{r}{2} + \frac{r}{2}$. This implies that $P$ is an $r$-bounded quasipartition of $M$ and proves the claim. Therefore ${\cal D}$ is a $2 \beta$-Lipschitz $\frac{1}{2n}$-forcing distribution over $r$-bounded quasipartitions of $M_W$. This concludes the proof of the theorem.
\end{proof}

Now we present methods to use quasipartitions for constructing embeddings of quasimetric spaces into quasiultrametric spaces and 0-1 quasimetric spaces.
The proof resembles the argument used in \cite{bartal1996probabilistic} for computing random embeddings of a metric space into a tree.

\begin{theorem}\label{embedding into quasiultrametrics}
Let $M=(X,d)$ be an $n$-point quasimetric space and let $\beta>0$. Suppose that for any $r>0$, there exists a $\beta$-Lipschitz distribution over $r$-bounded quasipartitions of $M$. Then $M$ admits a random embedding into quasiultrametrics with distortion $O(\beta  \log{n})$.
\end{theorem}

\begin{proof}
We may assume w.l.o.g.~that the minimum distance in $M$ is 1 and the diameter is $\Delta$. By Lemma \ref{epsilon forcing distributions} it follows that for any $r>0$ there exists a $2\beta$-Lipschitz $\frac{1}{2n}$-forcing distribution ${\cal D}_{r}$ over $r$-bounded quasipartitions of $M$. To get the required embedding we combine a series of quasipartitions.
Let $S = \{P_{0},P_{1},\ldots,P_{\lf \log {\Delta} \rf } \}$ where $P_{i} \in {\cal D}_{2^i}$ is a randomly chosen $2^i$-bounded quasipartition from ${\cal D}_{2^i}$.

We combine the quasipartitions as follows to get a quasiultrametric $M^*$:
\begin{description}
\item{Step 1:} Set $d_{M^{*}}(u,v) = 2^{\lf\log{\Delta}\rf + 1}$ for all $u,v \in V(M)$. Set $i = \lf\log{\Delta}\rf$.
\item{Step 2:} Set $d_{M^{*}}(u,v) = 2^i$ for all $(u,v) \in P_{i}$ if $d_{M^{*}}(u,v) = 2^{i+1}$. Decrease $i$ by 1. Repeat step 2 if $i \geq 0$.
\end{description}
We first argue that $M^*$ is a quasiultrametric.
To that end, consider any $u,v \in X$.
Let $j$ be the maximum value of $i$ such that $(u,v) \not\in P_i$. This implies that $d_M^*(u,v) = 2^{i+1}$. Consider any $w \in X$. It must be that either $(u,w) \not\in P_j$ or $(w,v) \not\in P_j$ because $(u,v) \not\in P_j$. This implies that either $d_M^*(u,w) = 2^{i+1}$ or $d_M^*(w,v) = 2^{i+1}$. So, for any $u,v,w \in X$ it must be that $d_M^*(u,v) \leq \max\{d_M^*(u,w),d_M^*(w,v)\}$.
This establishes that $M^*$ is a quasiultrametric.

Next we argue that $M^*$ is non-contracting.
Let us suppose that the claim is false and that $M^*$ is contracting. This means that for some $u,v \in X$ we have $d_{M^*}(u,v) < d_{M}(u,v)$. This means that in some iteration of Step 2 we set $d_{M^{*}}(u,v) = 2^i$ for some $i< \log{d_{M}(u,v)}$. This implies that $(u,v) \in P_i$ even though $d_{M}(u,v) >2^i$, which is a contradiction.

It remains to show that 
$M^*$ has expansion $O(\beta  \log{n})$.
Let $u,v\in X$.
We have
\begin{align*}
 \displaystyle \mathop{\mathbb{E}}\left[\frac{d_{M^*}(u,v)}{d_{M}(u,v)}\right] &\leq \sum_{i=0}^{\lf \log {\Delta}\rf} \Pr[(u,v) \not \in P_i]  \frac{2^{i+1}}{d_M(u,v)}\\
&\leq \sum_{i=0}^{\lf \log{ d_M(u,v) }\rf} \frac{2^{i+1}}{d_M(u,v)} + \sum_{i=\lf \log{ d_M(u,v) }\rf +1}^{\lf\log{ (2 n  d_M(u,v) )}\rf} 2\beta  \frac{d_M(u,v)}{2^i}  \frac{2^{i+1}}{d_M(u,v)}\\
& ~~~+ \sum_{i=\lf \log{ (2 n  d_M(u,v) )}\rf + 1}^{\lf \log {\Delta} \rf ]} 0\cdot \frac{2^{i+1}}{d_M(u,v)} \\
&\leq \frac{4 d_M(u,v)}{d_M(u,v)} + 4 \beta \log{n} =O(\beta \log {n}),
\end{align*}
concluding the proof.
\end{proof}

\begin{theorem}\label{embedding into 0-1 quasimetrics}
Let $M=(X,d)$ be an $n$-point quasimetric space and let $\beta>0$. Suppose that for any $r>0$, there exists a $\beta$-Lipschitz distribution over $r$-bounded quasipartitions of $M$. Then $M$ admits an embedding into a convex combination of 0-1 quasimetric spaces with distortion $O(\beta  \log{n})$.
\end{theorem}

\begin{proof}
We may assume w.l.o.g.~that the minimum distance in $M$ is 1 and the diameter is $\Delta$. By Lemma \ref{epsilon forcing distributions} it follows that for any $r >0$ there exists ${\cal D}_{r}$ a $\beta$-Lipschitz $\frac{1}{2n}$-forcing distribution over $r$-bounded quasipartitions of $M$. Let $S = \{ {\cal D}_{0},{\cal D}_{1},\ldots,{\cal D}_{\lf \log{\Delta} \rf}  \}$. Let $ c =  \sum_{i \in [0,\lf \log {\Delta}\rf]} 2^{i+1}$. Let $H$ be a discrete distribution over $S$ where the probability density function $F$ is given by $F({\cal D}_{i}) = \frac{2^{i+1}}{c}$.

Let us define $Y$ to be the event that a random quasipartition is selected from the distribution ${\cal D}_{i}$ where ${\cal D}_{i}$ is randomly chosen from the distribution $H$. This gives a distribution over a set of quasipartitions. We can replace every quasipartition $P$ in this set by a 0-1 quasimetric space $Q=(X,d_P)$ where for all $(u,v) \in P$ we have that $d_P(u,v)=0$ and for all $(u,v) \notin P$ we have that $d_P(u,v)=1$. This gives a distribution over a set of 0-1 quasimetric spaces which can be interpreted as a convex combination of 0-1 quasimetric spaces. Let the quasimetric space given by this convex combination of 0-1 quasimetric spaces be $\phi = (X,d_{\phi})$. We will now show that the distortion is bounded for this embedding. First we claim that for all $u,v \in X$, $d_{\phi}(u,v) \geq \frac{d(u,v)}{c}$. This can be shown as follows. We have that $$ d_{\phi}(u,v) = \sum_{i=0}^{\lf \log {\Delta}\rf} \frac{2^{i+1}}{c}  \Pr_{P\sim {\cal D}_{i}}[(u,v) \not \in P]
\geq \sum_{i=0}^{\lf \log{ d(u,v) }\rf} \frac{2^{i+1}}{c} 
\geq \frac{d(u,v)}{c},$$
which proves the claim. Next we show that for all $x,y \in X$, $d_{\phi}(x,y) \leq O(\beta \log{n}) \frac{d(x,y)}{c}$.
We have

\begin{align*}
 \displaystyle d_{\phi}(u,v) &= \sum_{i=0}^{\lf \log {\Delta}\rf} \frac{2^{i+1}}{c}  \Pr_{P\sim {\cal D}_{i}}[(u,v) \not \in P] \\
&\leq \sum_{i=0}^{\lf \log{ d(u,v) }\rf} \frac{2^{i+1}}{c} + \sum_{i=\lf \log{ d(u,v) }\rf +1}^{\lf\log{ (2 n  d(u,v) )}\rf} \beta  \frac{d(u,v)}{2^i}  \frac{2^{i+1}}{c} + \sum_{i=\lf \log{ (2 n  d(u,v) )}\rf + 1}^{\lf \log {\Delta} \rf } 0\cdot \frac{2^{i+1}}{c} \\
&\leq 4 \frac{d(x,y)}{c} + 2 \beta \log{n} \frac{d(x,y)}{c} \leq O(\beta \log{n}) \frac{d(x,y)}{c}. 
\end{align*}

From the above lower and upper bounds we get that $\phi$ is an embedding of $M$ with distortion $O(\beta  \log{n})$. This concludes the proof of the theorem.
\end{proof}

We get the following Corollaries by combining the above Theorems with the main result of Section \ref{sec:treewidth}.

\begin{corollary}
Let $M = (X,d)$ be the shortest path quasimetric space induced by a directed graph on $n$ vertices having treewidth $t$. Then $M$ admits a random embedding into a quasiultrametric space with distortion $O(t \log^2 n)$.
Moreover there exists an algorithm with running time polynomial in $n$ and $t$, that samples a random embedding into a quasiultrametric space with distortion $O(t \sqrt{\log t} \log^2 n)$.
\end{corollary}

\begin{proof}
The existential part follows from Lemma \ref{epsilon forcing distributions} and Theorems \ref{treewidth quasipartitions result} and \ref{embedding into quasiultrametrics}.
The computational part uses Theorem \ref{thm:tw poly algo}.
\end{proof}

\begin{corollary}\label{cor:bounded treewidth 0-1 quasimetrics}
Let $M = (X,d)$ be the shortest path quasimetric space induced by a directed graph on $n$ vertices having treewidth $t$. Then $M$ admits an embedding into a convex combination of 0-1 quasimetric spaces with distortion $O(t\log^2 (n))$.
Moreover, there exists an algorithm with running time polynomial in $n$ and $t$, that computes an embedding into a convex combination of 0-1 quasimetric spaces with distortion $O(t \sqrt{\log t} \log^2 n)$.
\end{corollary}

\begin{proof}
This follows from Lemma \ref{epsilon forcing distributions} and Theorem \ref{treewidth quasipartitions result} and \ref{embedding into 0-1 quasimetrics}.
The computational part uses Theorem \ref{thm:tw poly algo}.
\end{proof}

\subsection{Lower bounds}
We now obtain a lower bound on the distortion of random embeddings of general quasimetric spaces into quasiultrametric spaces. To this end we show that if a quasimetric space admits a random embedding into quasiultrametric spaces, then it is possible to construct a Lipschitz distribution over $r$-bounded quasipartitions of the quasimetric space for any $r>0$.

We also show that subdivided directed acyclic graphs cannot embed into a graph with bounded distortion unless there is a minor of the original graph present in the embedding.

\begin{theorem}\label{thm:q_lower}
Let $M=(X,d)$ be a quasimetric space and let $\beta>0$.
Suppose that $M$ admits a random embedding into quasiultrametric spaces ${\cal D^*}$ with distortion $\beta$. Then for any $r>0$, there exists an $\beta$-Lipschitz distribution over $r$-bounded quasipartitions of $M$. 
\end{theorem}

\begin{proof}
We get a $\beta$-Lipschitz distribution ${\cal D}$ over $r$-bounded quasipartitions of $M$ by modifying ${\cal D^*}$. For every quasiultrametric $M^* \in \cal D^*$ we add the $r$-bounded quasipartition $R(M^*)$ to ${\cal D}$ where $(u,v) \in R(M^*)$ iff $d_{M^*}(u,v) \leq r$. The probability of selecting any $R(M^*) \in {\cal D} $ is set to be equal to that of selecting the corresponding $M^* \in \cal D^*$. 

First we claim that $R(M^*)$ is an $r$-bounded quasipartition of $M$. Consider any $u,v,w \in V(M)$. Suppose we have that $(u,v) \in R(M^*)$ and $(v,w) \in R(M^*)$ then it must be that $d_{M^*}(u,v) \leq r$ and $d_{M^*}(v,w) \leq r$. Since $M^*$ is a quasiultrametric this implies that $d_{M^*}(u,w) \leq \max \{ d_{M^*}(u,v), d_{M^*}(v,w) \} \leq r$. This means that $(u,w) \in R(M^*)$ which implies that $R(M^*)$ is transitive and is hence a quasipartition of $M$. Since we have $(u,v) \in R(M^*)$ iff $d_{M^*}(u,v) \leq r$ and $d_{M^*}(u,v) \geq d_{M}(u,v)$ it follows that $R(M^*)$ is an $r$-bounded quasipartitions of $M$.

Next we prove that ${\cal D}$ is $\beta$-Lipschitz. We have
\begin{align*}
\displaystyle \beta &= \mathop{\mathbb{E}} \left[\frac{d_{M^*}(u,v)}{d_{M}(u,v)}\right] \geq \int_{d_M(u,v)}^{\infty} \left(\frac{z}{d_{M}(u,v)}\right)  \Pr [d_{M^*}(u,v) = z]  dz \\
& \geq  \left( \frac{r}{d_{M}(u,v)} \right)  \Pr [d_{M^*}(u,v) > r].\\
\end{align*} 
This implies that $\Pr [d_{M^*}(u,v) > r] \leq \beta \cdot \frac{d_{M}(u,v)}{r}$. Since $\Pr [d_{M^*}(u,v) > r] = \Pr[(u,v) \not \in R(M^*)]$, we have that ${\cal D}$ is $\beta$-Lipschitz and this concludes the proof of the theorem.
\end{proof}

The following Theorem follows directly from the work of Chuzhoy and Khanna \cite{directedcut} and a result (Theorem 4.2) of Charikar, Makarychev and Makarychev \cite{directedpartitioning}.

\begin{theorem}\label{thm:lowerbnd 0-1}
There exists a quasimetric space $M$ such that any embedding of $M$ into a convex combination of 0-1 quasimetric spaces has distortion $\Omega \left(\frac{n^{1/7}}{\log^{4/7}{n}}\right)$.
\end{theorem}

A lower bound on the flow-cut gap of Directed \spcut does not directly give the same lower bound on the quality of Lipschitz quasipartitions of quasimetric spaces. The following Theorem follows from the work of Chuzhoy and Khanna \cite{directedcut}.

\begin{theorem}\label{thm:lowerbnd}
There exists a quasimetric space $M$ and a positive $r$ such that any distribution over $r$-bounded quasipartitions of $M$ is $\Omega \left(\frac{n^{1/7}}{\log^{4/7}{n}}\right)$-Lipschitz.
\end{theorem}

\begin{proof}
We use a construction proposed by J. Chuzhoy and S. Khanna \cite{directedcut} for the Vertex Multicut problem. They construct a Multicut instance $G = (V,T,E)$ which is an unweighted directed acyclic graph where $V$ is the set of non-terminal vertices, $T$ is the set of source and sink vertices and $E$ the set of edges. This construction satisfies the following two properties:
\begin{description}
\item{Property 1:} Any path connecting a source vertex $s_i$ to its corresponding sink $t_i$ contains at least $L$ non-terminal vertices where $L=\Omega (\frac{n^{1/7}}{\log^{4/7}{n}})$ where $n = |V|$.
\item{Property 2:} At least $\Omega(n)$ non-terminal vertices must be removed to disconnect all source-sink pairs.
\end{description} 
We will construct a directed graph $G^* = (V^*,T,E^*)$ from $G$. Let $|V^*| = N$. Replace every non-terminal vertex $v \in V$ by $v^+$ and $v^-$ in $V^*$ and add the directed edge $(v^+,v^-)$ with unit weight in $E^*$. We will refer to these edges as \emph{weighted} edges. Clearly the number of \emph{weighted} edges is exactly $n$ and $N = |V^*| \leq 2 |V| = 2  n$. For any edge $(x,y) \in G$ add an edge $(x^-,y^+)$ with edge weight zero to $E^*$. We will refer to these edges as \emph{unweighted} edges.

Consider $M$, the shortest path quasimetric space induced by $G^*$ and let $r = L - \epsilon$. Let ${\cal D}$ be a $\beta$-Lipschitz distribution over $r$-bounded quasipartitions of $M$. Let $P \in {\cal D} $ be any $r$-bounded quasipartition of $M$. We know that for any \emph{unweighted} edge $(u^+,v^-) \in E^*$ there must be an ordered pair $(u^+,v^-) \in P$ because $d_M(u^+,v^-) = 0$ (Otherwise $\beta$ is $\infty$). Let the set of \emph{weighted} edges be $E_1 = \{(v_1^-,v_1^+),(v_2^-,v_2^+),\ldots,(v_n^-,v_n^+)\}$. Let $x_P : E_1 \rightarrow \{0,1\}$ be an indicator function. $x_P((v_i^-,v_i^+)) = 1$ if $(v_1^-,v_1^+) \not\in P$ otherwise $x_P((v_i^-,v_i^+)) = 0$. Since $d_M(s_i,t_i) \geq L$ for any source-sink pair $s_i,t_i \in T$ from Property 1, and any quasipartition is transitive, it must be that every source-sink pair is disconnected in $P$. This along with Property 2 gives, $$\sum_{e \in E_1} x_P(e) \geq \Omega (n), \, P \in {\cal D}$$
Which implies that, 
$$\displaystyle \mathop{\mathbb{E}}_{P\sim {\cal D}} \left[\sum_{e \in E_1} x_P(e)\right] \geq \Omega (n)$$
But taking union bound we get,
$$\displaystyle \mathop{\mathbb{E}}_{P\sim {\cal D}} \left[\sum_{e \in E_1} x_P(e) \right] \leq \sum_{e \in E_1}\Pr_{P\sim {\cal D}}[e \not \in P] $$
Therefore,
$$\sum_{e \in E_1}\Pr_{P\sim {\cal D}}[e \not \in P] \geq \Omega (n)$$
Since ${\cal D}$ is a $\beta$-Lipschitz distribution we have,
$\displaystyle \sum_{e \in E_1} \beta  \frac{d_M(e)}{r} \geq \Omega (n)$. Summing over all the \emph{weighted} edges we get that $n  \beta  \frac{1}{r} \geq \Omega (n) $. This implies that,
\begin{align*}
\beta &\geq \Omega (n) \frac{r}{n} \geq \Omega (L) \geq \Omega \left(\frac{n^{1/7}}{\log^{4/7}{n}}\right) \geq \Omega \left(\frac{N^{1/7}}{\log^{4/7}{N}}\right)
\end{align*}
This concludes the proof.
\end{proof}

Combining Theorems \ref{thm:q_lower} and \ref{thm:lowerbnd} gives the following Corollary.
\begin{corollary}

There exists an $n$-point quasimetric space $M = (X,d)$ such that any random embedding of $M$ into quasiultrametric spaces has distortion $\Omega\left(\frac{n^{1/7}}{\log{n}^{4/7} }\right)$.

\end{corollary}

\begin{theorem}\label{theo:dags}
Let $G = (E,V)$ be a directed acyclic graph. Let $G' = (E',V')$ be the graph obtained by subdividing each directed edge of $G$ into three edges in the same direction, i.e., for any directed edge $(u,v) \in E$, $E'$ contains the directed edges $(u,x_{uv}), (x_{uv},y_{uv}), (y_{uv},v)$ obtained by subdividing $(u,v)$. Suppose that there exists an embedding of $G'$ into some graph $H$ with bounded distortion. Then $H$ must contain $G$ as a minor.
\end{theorem}

\begin{proof}
First we observe that since $G$ is a DAG it must be that $G'$ is also a DAG by construction. We note that $V' = V \cup \{x_{uv} : (u,v) \in E\} \cup \{y_{uv} : (u,v) \in E\}$. Consider any pair of edges $(u,v)$ and $(w,z)$ in $E$. We have that either $d_G(v,w) \rightarrow \infty$ or $d_G(z,u) \rightarrow \infty$ because $G$ is a DAG. Now let $\phi : V' \rightarrow X$ be a bounded distortion embedding of $G'$ into a directed graph $H = (X,E'')$. For any $u,v \in E''$ we will use the notation $P(u,v)$ to denote the shortest directed path from $u$ to $v$ in $H$. We claim that for all $(u,v) \in E'$ there has to be a directed path from $\phi(u)$ to $\phi(v)$ in $H$. This follows from the fact that the embedding has bounded expansion.

Let $(a,b),(c,d)\in E$ and $(c,d) \neq (a,b)$. Then $P(\phi(x_{cd}),\phi(y_{cd}))$ and $P(\phi(x_{ab}),\phi(y_{ab}))$ do not intersect in $H$. We can prove this by contradiction. Suppose $P(\phi(x_{cd}),\phi(y_{cd}))$ and $P(\phi(x_{ab}),\phi(y_{ab}))$ do intersect in $H$ then $d_H(\phi(x_{cd}),\phi(y_{ab}))$ and $d_H(\phi(x_{ab}),\phi(y_{cd}))$ are bounded. Since the embedding $\phi$ has bounded contraction this implies that $d_{G'}(x_{cd},y_{ab})$ and $d_{G'}(x_{ab},y_{cd})$ are also bounded. So there exists a directed path from $x_{cd}$ to $y_{ab}$ in $G'$. However the only incoming edge to $y_{ab}$ is $(x_{ab},y_{ab})$. This implies that the directed path from  $x_{cd}$ to $y_{ab}$ in $G'$ goes through $x_{ab}$. Similarly it must be that there exists a directed from  $x_{ab}$ to $y_{cd}$ in $G'$ that goes through $x_{cd}$. The consequence of this is that there exists a directed path from $x_{cd}$ to $x_{ab}$ and a directed path from $x_{ab}$ to $x_{cd}$. This contradicts the fact that $G'$ is a DAG.

Let $(a,b),(c,d)\in E$ where $a \neq d$ and let $(c,d) \neq (a,b)$. Then $P(\phi(a),\phi(x_{ab}))$ and $P(\phi(y_{cd}),\phi(d))$ do not intersect in $H$. We can prove this by contradiction. Suppose $P(\phi(a),\phi(x_{ab}))$ and $P(\phi(y_{cd}),\phi(d))$ do intersect in $H$ then $d_H(\phi(a),\phi(d))$ and $d_H(\phi(y_{cd}),\phi(x_{ab}))$ are bounded. Since the embedding $\phi$ has bounded contraction this implies that $d_{G'}(a,d)$ and $d_{G'}(y_{cd},x_{ab})$ are also bounded. So there exists a directed path from $a$ to $d$ in $G'$. Similarly there exists a directed path from $y_{cd}$ to $x_{ab}$ in $G'$. However the only outgoing edge of $y_{cd}$ is $(y_{cd},d)$ and the only incoming edge of $x_{ab}$ is $(a,x_{ab})$. Therefore this directed path goes through $d$ and then $a$ implying that there is a directed path from $d$ to $a$. This contradicts the fact that $G'$ is a DAG since $a$ and $d$ are part of a directed cycle.

Let $(a,b),(c,d)\in E$ where $a \neq c$. Then $P(\phi(a),\phi(x_{ab}))$ and $P(\phi(c),\phi(x_{cd}))$ do not intersect in $H$. We can prove this by contradiction. Suppose $P(\phi(a),\phi(x_{ab}))$ and $P(\phi(c),\phi(x_{cd}))$ do intersect in $H$ then $d_H(\phi(a),\phi(x_{cd}))$ and $d_H(\phi(c),\phi(x_{ab}))$ are bounded. Since the embedding $\phi$ has bounded contraction this implies that $d_{G'}(a,x_{cd})$ and $d_{G'}(c,x_{ab})$ are also bounded. So there exists a directed path from $a$ to $x_{cd}$ in $G'$. Since the only incoming edge of $x_{cd}$ is $(c,x_{cd})$ this directed path goes through $c$. Consequently there is a directed path from $a$ to $c$. Similarly there exists a directed path from $c$ to $x_{ab}$ in $G'$. However the only incoming edge of $x_{ab}$ is $(a,x_{ab})$. Therefore this directed path goes through $a$ implying that there is a directed path from $c$ to $a$. This contradicts the fact that $G'$ is a DAG since $a$ and $c$ are part of a directed cycle.

Let $(a,b),(c,d)\in E$ where $b \neq d$. Then $P(\phi(y_{ab}),\phi(b))$ and $P(\phi(y_{cd}),\phi(d))$ do not intersect in $H$. We can prove this by contradiction. Suppose $P(\phi(y_{ab}),\phi(b))$ and $P(\phi(y_{cd}),\phi(d))$ do intersect in $H$ then $d_H(\phi(y_{ab}),\phi(d))$ and $d_H(\phi(y_{cd}),\phi(b))$ are bounded. Since the embedding $\phi$ has bounded contraction this implies that $d_{G'}(y_{ab},d)$ and $d_{G'}(y_{cd},b)$ are also bounded. So there exists a directed path from $y_{ab}$ to $d$ in $G'$. Since the only outgoing edge of $y_{ab}$ is $(y_{ab},b)$ this directed path goes through $b$. Consequently there is a directed path from $b$ to $d$. Similarly there exists a directed path from $y_{cd}$ to $b$ in $G'$. However the only outgoing edge of $y_{cd}$ is $(y_{cd},d)$. Therefore this directed path goes through $d$ implying that there is a directed path from $d$ to $b$. This contradicts the fact that $G'$ is a DAG since $b$ and $d$ are part of a directed cycle.

We are now ready to find a $G$ minor in $H$. First we delete all edges and vertices that are not part of some $P(\phi(u),\phi(v))$ where $(u,v) \in E'$. Then for all $(u,v) \in E$ we merge $P(\phi(u),\phi(x_{uv}))$ with $u$ and $P(\phi(y_{uv}),\phi(v))$ with $v$ using edge contractions. For all $u \in V$ we denote the resulting merged vertex in $H$ by $u^*$. Finally for all $(u,v) \in E$ we contract the edges of $P(\phi(x_{uv}),\phi(y_{uv}))$ until we are left with a single directed edge from $u^*$ to $v^*$. This does not cause any merged vertex $a^*$ to merge with another merged vertex $b^*$ since the paths of the form $P(\phi(x_{ab}),\phi(y_{ab}))$ are non-intersecting. For all $(u,v) \in E$ the resulting graph has a directed edge $(u^*,v^*)$. This gives us the required minor of $G$ and concludes the proof.
\end{proof}

\section{Applications to directed cut problems}

In this section we will use the results from the previous section, for embedding quasimetric spaces into convex combinations of 0-1 quasimetric spaces, to get approximation algorithms for some cut problems on directed graphs.

\paragraph*{Directed Non-Bipartite \spcut}

Let $G$ be a directed graph with non-negative capacities $c(e)$ for all $e \in E(G)$. Let $T$ be a set of terminal vertex pairs $\{(s_1,t_1),(s_2,t_2),\ldots,(s_k,t_k)\}$. Let $\dem(i)$ be a non-negative demand for the terminal pair $(s_i,t_i)$. A \emph{cut} of $G$ is a set of edges $S \subset E(G)$. We define the \emph{capacity} of the cut $S$ to be $c(S) = \sum_{e \in S} c(e)$. Let $I_S$ be the set of all integers $i\in \{1,\ldots,k\}$ such that all paths from $s_i$ to $t_i$ have at least one edge in $S$. We define the demand separated by the cut $S$ to be $\dem(S) = \sum_{i\in I_S} \dem(i)$. The \emph{sparsity} of a cut $S$ is defined to be $c(S)/\dem(S)$. The goal of the Directed Non-Bipartite \spcut problem is to find the cut with minimum sparsity.

For non-uniform demands, there exists a $O(\sqrt{n})$-approximation by \cite{aprxdirectedsparsestcut}, which has been improved to $\tilde{O}(n^{11/23})$-approximation \cite{bestsparsestmulticut}. 
There is also a $2^{\Omega(\log^{1 - \epsilon}{n})}$-hardness due to Chuzhoy and Khanna \cite{directedcut}.

Consider the following standard LP relaxation of the Directed Non-Bipartite \spcut Problem on $G$.
\begin{align*}
 \displaystyle \min & \sum_{e \in E(G)} c(e)x(e)\\
&\sum_{(s_i,t_i) \in T} \dem(i)d(s_i,t_i) \geq 1\\
& x(e) \geq 0 \qquad \forall e \in E(G)\\
& d(u,v) \geq 0 \qquad \forall u,v \in V(G)
\end{align*}

In the LP, the $x(e)$ values can be treated as distance assignments for the edges. The $d(u,v)$ values are the shortest path distances from $u$ to $v$ in $G$ for edge weights defined by the distance assignments. Since $d$ is the shortest path distance, it follows that for all $(u,v) \in E$ we have that $d(u,v) \leq x(e)$ where $e = (u,v)$. Since replacing $x(e)$ with $d(u,v)$ only reduces the objective function while preserving feasibility, we have that the optimal edge distance assignment of the LP is given by a shortest path quasimetric space on $G$.

Charikar, Makarychev and Makarychev \cite{directedpartitioning} showed that the integrality gap of this LP is closely related to the minimum distortion achievable for embedding a quasimetric space into a convex combination of 0-1 quasimetric spaces. Theorem A.1 from their paper implies that the integrality gap of the LP for a graph $G$ with edge capacities $c(e)$ is equal to the minimum distortion for embedding a shortest path quasimetric space supported on $G$ into a convex combination of 0-1 quasimetric spaces (referred to as 0-1 semimetrics in their paper). We now describe how the minimum distortion embedding of a quasimetric space into a convex combination of 0-1 quasimetric spaces can be used to upper bound the integrality gap of the LP. First we observe that a 0-1 quasimetric space on $V(G)$ corresponds to a cut of $G$. This is because a 0-1 quasimetric space $M = (V(G),d^*)$ can be used to describe a cut $S \subset E(G)$ where $(u,v) \in S$ iff $d^*(u,v) = 1$ for all $(u,v) \in E(G)$. Let the optimal \spcut value for the LP be $LP_{\OPT}$. Let the shortest path quasimetric space on $G$ that gives the optimal solution for the LP be $M_{\OPT} = (V(G),d_{\OPT})$. We have that,
\begin{align*} 
\displaystyle \frac{\sum_{e=(u,v) \in E(G)}c(e)d_{\OPT}(u,v)}{\sum_{i \in [0,k]}\dem(i)d_{\OPT}(s_i,t_i)} = LP_{\OPT} .\\
\end{align*} 
Let $ \phi = \sum_{ j \in [0,m]} \alpha_{j} M_j $ be an $O(\log^2 (n))$ embedding of $M_{\OPT}$ into a convex combination of 0-1 quasimetric spaces, where $M_j = (V(G),d_j)$ is a 0-1 quasimetric space and $\alpha_{j}$ is a non-negative real number for any $j \in [0,m]$. We have that,
\begin{align*}
O(\log^2(n)) LP_{\OPT} &\geq \frac{\sum_{e=(u,v) \in E(G)}c(e)\sum_{ j \in [0,m]} \alpha_{j}d_{j}(u,v)}{\sum_{i \in [0,k]}\dem(i)\sum_{ j \in [0,m]} \alpha_{j}d_{j}(s_i,t_i)} \\
&\geq \frac{\sum_{ j \in [0,m]} \alpha_{j} \sum_{e=(u,v) \in E(G)}c(e)d_{j}(u,v)}{\sum_{ j \in [0,m]} \alpha_{j}\sum_{i \in [0,k]}\dem(i)d_{j}(s_i,t_i)}\\
&\geq \min_{j \in [0,m]}{\frac{\sum_{e=(u,v) \in E(G)}c(e)d_{j}(u,v)}{\sum_{i \in [0,k]}\dem(i)d_{j}(s_i,t_i)}}.\\
\end{align*}
Therefore the cut corresponding to the 0-1 quasimetric space in $\phi$ having minimum sparsity gives an integral solution to the LP that is at most an $O(\log^2(n))$ factor larger than $LP_{\OPT}$. Combined with Corollary  \ref{cor:bounded treewidth 0-1 quasimetrics} and Theorem \ref{thm:tw poly algo} this implies the following corollary.


\begin{corollary}\label{cortw}
The integrality gap (which is also the flow-cut gap) of the Directed Non-Bipartite  \spcut LP relaxation on graphs of $n$ vertices and treewidth $t$ is $O(t\log^2 n)$. Moreover there exists a polynomial-time $O(t  \sqrt{\log t} \log^2 n)$-approximation algorithm for the Directed Non-Bipartite \spcut problem on such graphs with running time linear in $t$ and polynomial in $n$.
If a tree decomposition of width $t$ is given as part of the input, then there exists a polynomial-time $O(t \log^2 n)$-approximation algorithm.
\end{corollary}

\paragraph*{Directed Multicut}
Let $G$ be a directed graph with non-negative capacities $c(e)$ for all $e \in E(G)$. Let $T$ be a set of terminal vertex pairs $\{(s_1,t_1),(s_2,t_2),\ldots,(s_k,t_k)\}$. The goal of the Directed Multicut problem is to find the minimum capacity cut that separates all terminal vertex pairs.

There is a $\tilde{O}(n^{2/3}/\OPT^{1/3})$-approximation by Kortsarts, Kortsarz and Nutov \cite{aprxdirectedmulticut3}, and a $\tilde{O}(n^{11/23})$-approximation by Agarwal, Alon and Charikar \cite{bestsparsestmulticut}. 
Finally, there is a $2^{\Omega(\log^{1 - \epsilon}{n})}$-hardness due to Chuzhoy and Khanna \cite{directedcut}.

It is known that, for all $\beta>0$, if there exists a polynomial time $\beta$-approximation algorithm for Directed Multicut then there also exists a  polynomial time $O(\beta \log n)$-approximation for Directed Non-Bipartite \spcut (Section 7 of \cite{bestsparsestmulticut}).
Moreover, for all $\alpha>0$, if there exists a polynomial time $\alpha$-approximation algorithm for Directed Non-Bipartite \spcut then there also exists a polynomial time $O(\alpha \log n)$-approximation for Directed Multicut.

Given an $\alpha$-approximation for the Directed Non-Bipartite \spcut Problem we can find an $(\alpha \log{n})$-approximation for the Directed Multicut Problem. The following algorithm gives us the required multicut.

\begin{algorithm}[htbp]
\caption{Reduction from Directed Multicut to Directed Non-bipartite Sparsest-Cut}
\begin{algorithmic}
\Require A directed graph $G$ with non-negative capacities $c(e)$ for all $e \in E(G)$. A set of terminal vertex pairs $T = \{(s_1,t_1),(s_2,t_2),\ldots,(s_k,t_k)\}$.
\Ensure An $(\alpha \log{n})$-approximation for the Directed Multicut Problem on $G$.\\
Let $T'$ be the current set of terminal vertices to be separated; initially $T' = T$. Let $i=1$.
\begin{description}
\item{Step 1.} Find $J_i$, an $\alpha$-approximation for the Directed Non-Bipartite  \spcut Problem on $G$, where the set of terminal pairs is given by $T'$, and the demand for all terminal pairs is set to $1$. Let $S_i \subset T'$ be the set of terminal pairs separated by the cut. Set $T' = T' - S_i$. Let $i=i+1$.
\item{Step 2.} If $T' \neq \emptyset$ repeat Step 1.
\end{description}
\end{algorithmic}
\end{algorithm}

The union of the directed non-bipartite cuts $\displaystyle \cup_{i} J_i$ gives the required directed multicut. Since the algorithm only terminates when all terminal pairs are separated it is indeed a multicut.

We can now upper bound the capacity of the cut. Let the capacity of the optimal multicut be $\OPT$. Since we use an $\alpha$-approximation for Directed \spcut we have that $\frac{c(J_i)}{|S_i|} \leq \alpha \frac{\\OPT}{|T'|}$ at every iteration of Step 1. This implies that
\begin{align*}
\sum_{i} c(J_i) &\leq \alpha \OPT\left( \frac{|S_1|}{|T|} + \frac{|S_2|}{|T-S_1|} + \frac{|S_3|}{|T-S_1-S_2|} + \ldots + \frac{|S_m|}{|S_m|} \right) \\
 &\leq O(\alpha \log{n})\OPT.
\end{align*}

\begin{corollary}
There exists an $O(t  \sqrt{\log t} \log^3 n)$-approximation algorithm for the Directed Multicut problem on $n$-vertex graphs of treewidth $t$, with running time polynomial in both $n$ and $t$.
If a tree decomposition of width $t$ is given as part of the input, then there exists a polynomial-time $O(t \log^3 n)$-approximation algorithm.
\end{corollary}


\newpage
\bibliography{bibfile}

\begin{thebibliography}{10}
\providecommand{\url}[1]{{#1}}
\providecommand{\urlprefix}{URL }
\expandafter\ifx\csname urlstyle\endcsname\relax
  \providecommand{\doi}[1]{DOI~\discretionary{}{}{}#1}\else
  \providecommand{\doi}{DOI~\discretionary{}{}{}\begingroup
  \urlstyle{rm}\Url}\fi

\bibitem{abraham2008nearly}
Abraham, I., Bartal, Y., Neiman, O.: Nearly tight low stretch spanning trees.
\newblock arXiv preprint arXiv:0808.2017  (2008)

\bibitem{bestsparsestmulticut}
Agarwal, A., Alon, N., Charikar, M.S.: Improved approximation for directed cut
  problems.
\newblock In: Proceedings of the thirty-ninth annual ACM symposium on Theory of
  computing, pp. 671--680. ACM (2007)

\bibitem{arora2008euclidean}
Arora, S., Lee, J., Naor, A.: Euclidean distortion and the sparsest cut.
\newblock Journal of the American Mathematical Society \textbf{21}(1), 1--21
  (2008)

\bibitem{bartal1996probabilistic}
Bartal, Y.: Probabilistic approximation of metric spaces and its algorithmic
  applications.
\newblock In: Foundations of Computer Science, 1996. Proceedings., 37th Annual
  Symposium on, pp. 184--193. IEEE (1996)

\bibitem{bartal1998approximating}
Bartal, Y.: On approximating arbitrary metrices by tree metrics.
\newblock In: Proceedings of the thirtieth annual ACM symposium on Theory of
  computing, pp. 161--168. ACM (1998)

\bibitem{borradaile2009randomly}
Borradaile, G., Lee, J.R., Sidiropoulos, A.: Randomly removing g handles at
  once.
\newblock In: Proceedings of the twenty-fifth annual symposium on Computational
  geometry, pp. 371--376. ACM (2009)

\bibitem{calinescu2005approximation}
Calinescu, G., Karloff, H., Rabani, Y.: Approximation algorithms for the
  0-extension problem.
\newblock SIAM Journal on Computing \textbf{34}(2), 358--372 (2005)

\bibitem{directedpartitioning}
Charikar, M., Makarychev, K., Makarychev, Y.: Directed metrics and directed
  graph partitioning problems.
\newblock In: Proceedings of the seventeenth annual ACM-SIAM symposium on
  Discrete algorithm, pp. 51--60. Society for Industrial and Applied
  Mathematics (2006)

\bibitem{directedcut}
Chuzhoy, J., Khanna, S.: Polynomial flow-cut gaps and hardness of directed cut
  problems.
\newblock Journal of the ACM (JACM) \textbf{56}(2), 6 (2009)

\bibitem{fakcharoenphol2003tight}
Fakcharoenphol, J., Rao, S., Talwar, K.: A tight bound on approximating
  arbitrary metrics by tree metrics.
\newblock In: Proceedings of the thirty-fifth annual ACM symposium on Theory of
  computing, pp. 448--455. ACM (2003)

\bibitem{feige2008improved}
Feige, U., Hajiaghayi, M., Lee, J.R.: Improved approximation algorithms for
  minimum weight vertex separators.
\newblock SIAM Journal on Computing \textbf{38}(2), 629--657 (2008)

\bibitem{aprxdirectedsparsestcut}
Hajiaghayi, M.T., R{\"a}cke, H.: An-approximation algorithm for directed
  sparsest cut.
\newblock Information Processing Letters \textbf{97}(4), 156--160 (2006)

\bibitem{indyk20048}
Indyk, P., Matou{\v{s}}ek, J.: Low-distortion embeddings of finite metric
  spaces.
\newblock Handbook of Discrete and Computational Geometry p. 177 (2004)

\bibitem{indyk2007probabilistic}
Indyk, P., Sidiropoulos, A.: Probabilistic embeddings of bounded genus graphs
  into planar graphs.
\newblock In: Proceedings of the twenty-third annual symposium on Computational
  geometry, pp. 204--209. ACM (2007)

\bibitem{kelner2011metric}
Kelner, J.A., Lee, J.R., Price, G.N., Teng, S.H.: Metric uniformization and
  spectral bounds for graphs.
\newblock Geometric and Functional Analysis \textbf{21}(5), 1117--1143 (2011)

\bibitem{klein1993excluded}
Klein, P., Plotkin, S.A., Rao, S.: Excluded minors, network decomposition, and
  multicommodity flow.
\newblock In: Proceedings of the twenty-fifth annual ACM symposium on Theory of
  computing, pp. 682--690. ACM (1993)

\bibitem{aprxdirectedmulticut3}
Kortsarts, Y., Kortsarz, G., Nutov, Z.: Greedy approximation algorithms for
  directed multicuts.
\newblock Networks \textbf{45}(4), 214--217 (2005)

\bibitem{krauthgamer2005measured}
Krauthgamer, R., Lee, J.R., Mendel, M., Naor, A.: Measured descent: A new
  embedding method for finite metrics.
\newblock Geometric \& Functional Analysis GAFA \textbf{15}(4), 839--858 (2005)

\bibitem{lee2005extending}
Lee, J.R., Naor, A.: Extending lipschitz functions via random metric
  partitions.
\newblock Inventiones mathematicae \textbf{160}(1), 59--95 (2005)

\bibitem{lee2009geometry}
Lee, J.R., Sidiropoulos, A.: On the geometry of graphs with a forbidden minor.
\newblock In: Proceedings of the forty-first annual ACM symposium on Theory of
  computing, pp. 245--254. ACM (2009)

\bibitem{lee2010genus}
Lee, J.R., Sidiropoulos, A.: Genus and the geometry of the cut graph.
\newblock In: Proceedings of the twenty-first annual ACM-SIAM symposium on
  Discrete Algorithms, pp. 193--201. Society for Industrial and Applied
  Mathematics (2010)

\bibitem{linial1995geometry}
Linial, N., London, E., Rabinovich, Y.: The geometry of graphs and some of its
  algorithmic applications.
\newblock Combinatorica \textbf{15}(2), 215--245 (1995)

\bibitem{sidiropoulos2010optimal}
Sidiropoulos, A.: Optimal stochastic planarization.
\newblock In: Foundations of Computer Science (FOCS), 2010 51st Annual IEEE
  Symposium on, pp. 163--170. IEEE (2010)

\end{thebibliography}



\end{document}